\documentclass[10pt, conference, compsocconf]{IEEEtran}

\usepackage{makeidx}
\usepackage{graphicx,comment}
\usepackage{amssymb}
\usepackage{enumerate}
\usepackage{multirow}
\usepackage{amsmath}
\usepackage{amsthm}
\usepackage{tabularx}
\usepackage{dsfont}
\usepackage{extarrows}
\usepackage{epstopdf}
\usepackage[]{algorithm2e}
\usepackage{xcolor}
\usepackage{fixltx2e}  
\usepackage{picinpar} 
\usepackage{booktabs,threeparttable}
\usepackage{mdwlist}

\begin{document}
%
\title{Modelling and Analysis of Network Security\\
- an Algebraic Approach}

\author{\IEEEauthorblockN{Qian Zhang}
\IEEEauthorblockA{Institute of Software \\Chinese Academy of Sciences, CAS\\
Beijing China\\
Email:zhangq@ios.ac.cn}
\and
\IEEEauthorblockN{Ying Jiang}
\IEEEauthorblockA{Institute of Software \\Chinese Academy of Sciences, CAS\\
Beijing China\\
Email:jy@ios.ac.cn}
\and
\IEEEauthorblockN{Peng Wu}
\IEEEauthorblockA{Institute of Software \\Chinese Academy of Sciences, CAS\\
Beijing China\\
Email:wp@ios.ac.cn}}
\maketitle

\begin{abstract}
Game theory has been applied to investigate network security. But different security scenarios were often modeled via different types of games and analyzed in an ad-hoc manner. In this paper, we propose an algebraic approach for modeling and analyzing uniformly several types of network security games.
This approach is based on a probabilistic extension of the value-passing Calculus of Communicating Systems (CCS), which is a common formal language for modeling concurrent systems. Our approach gives a uniform security model for different security scenarios.
We present then a uniform algorithm for computing the Nash equilibria strategies on this security model.
In a nutshell, the algorithm first generates a network state transition graph for our security model, then simplifies this transition graph through graph-theoretic abstraction and bisimulation minimization.
Then, a backward induction method, which is only applicable to finite tree models, can be used to compute all the Nash equilibria strategies of the (possibly infinite) security models.
This algorithm is implemented and can be tuned smoothly for computing its social optimal strategies, and its termination and correctness are proved.
The effectiveness and efficiency of this approach are demonstrated with two detailed examples from the field of network security.
\end{abstract}


\IEEEpeerreviewmaketitle

\section{Introduction}
As the Internet has become ubiquitous, the risk posed by network attacks has greatly increased. 
Generally, network security scenarios can be classified into two main categories: one in which defenders have full understanding of malicious levels of users via white list or black list, and the other one in which defenders have no accurate knowledge of the users's types.
How to devise effective defense mechanisms against various attacks is a fundamental research area.
A Nash Equilibrium Strategy (NES) \cite{david}\cite{rd} defines a relative optimal defense mechanism, where neither attackers nor defenders are willing to change their current offensive-defensive behaviors.

In recent two decades, game-theoretic approaches have been applied to investigate network security \cite{xiannuan}\cite{sankardas}\cite{syverson}.
To name a few, complete information games \cite{martin}, in which each player knows the types, strategies and payoffs of all the other players, can be applied to modelling the security scenarios in which defenders know the users' types \cite{nguyen}\cite{klye}\cite{ls}\cite{jean}\cite{xiaolin}.
While the incomplete information games can be used to model the scenarios in which defenders have no idea of users' type  \cite{Harsanyi}\cite{patcha}\cite{monireh}\cite{nguyen2}\cite{karel}\cite{chaozhang}.
However, specific game models are only suitable for analyzing NESs under specific security scenarios.
How to find NESs for different security scenarios with a uniform framework is far from having been solved.

CCS is a common formal language for modeling concurrent systems.
It can describe interactive behaviors vividly up to its interleaving semantics.
Inspired by the generative model for probabilistic CCS \cite{rob}, we propose a generative probabilistic extension for the value-passing CCS (PVCCS\textsubscript{G} for short), and then a uniform security model based on PVCCS\textsubscript{G} is put forward.
We are the first, to our knowledge, to present a uniform framework for analyzing the NESs for various network security scenarios.

For a network security scenario with one user (a legitimate user or an attacker) and one defender as participants, as the defender does not know the user's type which means the user's maliciousness, we introduce another virtual participant ``Nature" to perform the Harsanyi transformation \cite{rd}, i.e., to convert nondeterministic choices under uncertain user types to quantitative choices of risk conditions.
Our approach interprets the network security scenario as a state transition system.
The states depend on the behaviors of the participants.
The state transitions depend on the interactions among the participants.
We then present a uniform algorithm to compute all the NESs for different network security scenarios automatically.
Firstly, we minimize the PVCCS\textsubscript{G} based security model up to probabilistic bisimularity, which is a well-defined technique in process calculi. In this way, the semantically equivalent states can be unified as single ones.
Then we abstract the minimized model in a graph-theoretic manner.
The abstracted model is then converted to a finite hierarchical graph by Tarjan's algorithm \cite{reinhard} to increase reusability and parallelization.
Finally, we compute the NESs backward inductively in the hierarchical graph.
We take two different security scenarios from \cite{harkeerat} and \cite{klye} for case studies. The experimental results are rather promising in terms of the effectiveness and flexibility of our approach.

The major contributions of our work are as follows.
\begin{itemize}
\item
We propose a uniform framework based on PVCCS\textsubscript{G} to characterize the security scenarios modeled via complete or incomplete information games, which general game-theoretic approaches cannot support yet.
\item
We minimize the PVCCS\textsubscript{G} based security model by probabilistic bisimularity and abstract the minimized model by graph theoretic methods. It reduces the state space and makes our model to be scalable.
\item
We propose a uniform algorithm to compute out the NESs for various security scenarios automatically.
The efficiency of the algorithm benefits from high reusability, parallelization and the minimized model.
\item
We filter out an invalid NES from the results obtained by classical game-theoretic approach \cite{klye}. It is an incredible threat \cite{martin} which is a NES but will never happen in the real situation if the players are rational.
\end{itemize}

The rest of the paper is organised as follows.
We establish a generative probabilistic extension of the value-passing CCS (PVCCS\textsubscript{G}) and construct a PVCCS\textsubscript{G} based security model (Section 2);
give the formal definition of NES in this model and present the algorithm (Section 3);
illustrate the efficiency of our method by two security scenarios (Section 4);
finally, discuss the conclusions (Section 5).

\section{Modelling based on PVCCS\textsubscript{G}}
\newtheorem{definition}{Definition}[section]
\newtheorem{Theorem}{Theorem}[section]
\newtheorem{Lemma}{Lemma}[section]

\subsection{ PVCCS\textsubscript{G}}
Inspired by the generative model for probabilistic CCS \cite{rob}, we propose a generative model for probabilistic value-passing CCS (PVCCS\textsubscript{G}).
\paragraph{Syntax}
Let $\mathbf{\mathcal{A}}$ be a set of channel names, and $a$ range over $\mathbf{\mathcal{A}}$, and $\mathbf{\mathcal{\overline{A}}}$ be the set of co-names, i.e., $\mathbf{\mathcal{\overline{A}}}=\{\overline{a}\mid a\in \mathbf{\mathcal{A}}\}$.
Let $Label=\mathcal{A} \cup \mathcal{\overline{A}}$, ${\it Var}$ be a set of value variables, and $x$ range over ${\it Var}$.   $\mathit{Val}$ is a value set, and $v$ range over $\mathit{Val}$.
$\sf e$ and $\sf b$ denote a value expression and a boolean expression, respectively.
Let $Act$ be a set of actions, and $\alpha$ range over $Act$. $Act=\{a(x)\mid a\in \mathcal{A}\}\cup \{\overline{a}(\sf e)\mid \overline{\it a}\in \mathcal{\overline{A}}\}\cup \{\tau \}$, where $\tau$ is the invisible action, $a(x)$ and $\overline{a}(\sf e)$ denote an input prefix action and an output prefix action, respectively.

Let $\mathbf{Pr}\textsubscript{G}$ be the set of processes in PVCCS\textsubscript{G}. Each process expression $E$ is defined inductively as follows:
\begin{equation*}
\begin{aligned}
E::=&Nil\mid \alpha.E \mid\underset{i\in I}{\sum}[p_{\textit{i}}]E_{\textit{i}}\mid E_1|E_2\mid E\backslash R\mid E[f]\mid \\ &\textrm{{\bf if}}~{\sf b}~{\rm {\bf then}}~E_1~{\rm {\bf else}}~E_2
\mid {\it A(x)}\\
\alpha::=&a(x)\mid \overline{a}(\sf e)
\end{aligned}
\end{equation*}
$Nil$ is the empty process which does nothing.
$\alpha.E$ is a prefixing process which evolves to $E$ by performing $\alpha$.
$\underset{i\in I}{\sum}[p_{\textit{i}}]E_{\textit{i}}$ is a probabilistic choice process which means $E_{\textit{i}}$ will be chosen with probability $p_{\textit{i}}$, where $I$ is an index set, and for $\forall i\in I$, $p_{\textit{i}}\in (0,1]$, $\underset{i\in I}{\dot{\sum}}p_{\textit{i}}=1$. $\sum$ and $\dot{\sum}$ are summation notations for processes and real numbers, respectively.
$E_{1}|E_{2}$ represents the combined behavior of $E_{1}$ and $E_{2}$ in parallel.
$E\backslash R$ is a process with channel restriction, whose behavior is like that of $E$ as long as $E$ does not perform any action with channel $a\in R\cup \overline{R}$, $R\subseteq \mathcal{A}$.
$E[f]$ means relabeling the channels of process $E$ as indicated by $f$, where $f:Label\rightarrow Label$ is a relabeling function.
${\rm {\bf if}}~{\sf b}~{\rm {\bf then}}~E_1~{\rm {\bf else}}~E_2$ is a conditional process which enacts $E_{1}$ if ${\sf b}$ is $true$, else $E_{2}$.
Each process constant $A(x)$ is defined recursively as $A(x)\stackrel{\textit{def}}{=}E$, where $E$ contains no process variables and no free value variables except $x$.

\paragraph{Semantics}
The semantics of PVCCS\textsubscript{G} are defined in Table \ref{operationalsemanticspvccsG}. $E\stackrel{\alpha[p]}{\rightarrow}E'$ means that, by performing an action $\alpha$, $E$ will evolve to $E'$ with probability $p$.
Let ${\sf chan}:Act\rightarrow \mathcal{A}$, i.e., ${\sf chan}(a(x))={\sf chan}(\overline{a}({\sf e}))=a$. $E\{{\sf e}/x\}$ means substituting with ${\sf e}$ for every free occurrences of $x$ in process $E$.
Let $\nu(E,R)=\dot{\sum}\{p_{\textit{i}}\mid E\stackrel{\alpha[p_{\textit{i}}]}{\longrightarrow}E_{\textit{i}}, ~{\sf chan}(\alpha)\notin R\}$ and $\wp$ be the powerset operator.
$\mathbf{Pr\textsubscript{G}}/\mathcal{R}$ denotes the set of equivalence classes induced by an equivalence relation $\mathcal{R}$ over $\mathbf{Pr\textsubscript{G}}$.
\begin{definition}
Let $\mu:(\mathbf{Pr\textsubscript{\rm G}}\times Act\times \wp(\mathbf{Pr\textsubscript{\rm G}}))\rightarrow[0,1]$ is a total function given by: $\forall \alpha\in Act$, $\forall E\in \mathbf{Pr\textsubscript{\rm G}}$, $\forall C\subseteq \mathbf{Pr\textsubscript{\rm G}}$, $\mu(E,\alpha,C)=\dot{\sum}\{p|E\stackrel{\alpha[p]}{\longrightarrow}E',~E'\in C\}$. \end{definition}

\begin{definition}
An equivalence relation $\mathcal{R}\subseteq \mathbf{Pr}\textsubscript{\rm G} \times \mathbf{Pr}\textsubscript{\rm G}$ is a \textit{probabilistic bisimulation} if $(E, E')\in \mathcal{R}$ implies: $\forall C\in \mathbf{Pr}\textsubscript{\rm G}/\mathcal{R}$, $\forall\alpha\in Act$, $\mu(E, \alpha, C)=\mu(E', \alpha, C)$.
\end{definition}

$E$ and $E'$ are probabilistic bisimilar, written as $E\sim E'$, if there exists a probabilistic bisimulation $\mathcal{R}$ s.t. $E\mathcal{R}E'$.
\begin{table*}
\centering
\renewcommand{\arraystretch}{1.5}
\caption{\label{operationalsemanticspvccsG}Operational semantics of PVCCS\textsubscript{G}}
\begin{tabular}{ll}\hline
$[{\rm In}]\frac{}{a(x).E\stackrel{a({\sf e})[1]} {\longrightarrow }E\{{\sf e}/x\}}$&
$[{\rm Out}]\frac{}{\overline{a}({\sf e}).E \stackrel{\overline{a}({\sf e})[1]}{\longrightarrow}E} $\\
$[{\rm Sum}] \frac{E_{\textit{i}}\stackrel{\alpha[q]}{\longrightarrow}E_{\textit{i}}^{'} }{\underset{i\in I}{\sum}[p_{\textit{i}}]E_{\textit{i}}\stackrel{\alpha[p_{\textit{i}}\cdot q]}{\longrightarrow}E_{\textit{i}}^{'}} $& $[{\rm Res}] \frac{E\stackrel{\alpha[p]}{\longrightarrow}E'}{E\backslash R\stackrel{\alpha[p/r]}{\longrightarrow}E'\backslash R} ~~~({\sf chan}(\alpha)\notin R, r=\nu(E,R)) $\\
$[{\rm {Par_l}}] \frac{E_1\stackrel{\alpha[p]}{\longrightarrow}E_1^{'}}{E_1|E_2\stackrel{\alpha [p]}{\longrightarrow}E_1^{'}|E_2}$&
$[{\rm Par_r}] \frac{E_2\stackrel{\alpha[p]}{\longrightarrow}E_2^{'}}{E_1|E_2\stackrel{\alpha [p]}{\longrightarrow}E_1|E_2^{'}} $ ~~~~~
$[{\rm Rel}] \frac{E\stackrel{\alpha[p]}{\longrightarrow}E'}{E[f]\stackrel{f(\alpha)[p]}{\longrightarrow}E'[f]}$\\
$[{\rm Com}]\frac{E_1\stackrel{a({\sf e})[p]}{\longrightarrow}E_1^{'},~ E_2\stackrel {\overline{a}({\sf e})[q]}{\longrightarrow}E_2^{'}}{E_1|E_2\stackrel{\tau[p\cdot q]}{\longrightarrow}E_1^{'}|E_2^{'}}$ & $[{\rm Con}]\frac{E\{{\sf e}/x\}\stackrel{\alpha[p]}{\longrightarrow} E'}{A({\sf e})\stackrel{\alpha[p]}{\longrightarrow}E'}~~(A(x)\stackrel{\textit{def}} {=}E)$\\
$[{\rm if_t}]\frac{E_{1}\stackrel{\alpha[p]}{\longrightarrow}E_{1}^{'}} {{\rm {\bf if}}~{\sf b}~{\rm {\bf then}}~E_1~{\rm {\bf else}}~E_2\stackrel{\alpha[p]} {\longrightarrow}E_{1}^{'}}({\sf b}=true)  $ & $ [{\rm if_f}]\frac{E_{2}\stackrel{\alpha[p]}{\longrightarrow}E_{2}^{'}} {{\rm {\bf if}}~{\sf b}~{\rm {\bf then}}~E_1~{\rm {\bf else}}~E_2\stackrel{\alpha[p]} {\longrightarrow}E_{2}^{'}}({\sf b}=\textit{false})  $\\
\hline
\end{tabular}
\end{table*}
\subsection{PVCCS\textsubscript{G} based Security Model}
A network system can be abstracted as four participants: the Nature, one user, one defender and the network environment which is the hardware and software services of the network under consideration.
To construct the PVCCS\textsubscript{G} based security model, the following aspects are addressed.
\begin{enumerate}[1]
\item $\textit{Ty}$: the type set of the user, and $t$ range over ${\it Ty}$.
\item $S$ : the set of network states, and $s$ range over $S$.
\item $A^{u}$ and $A^{d}$: the action sets of the user and the defender, respectively. Let $A^{u}=\underset{s\in S,t\in \textit{Ty}}{\cup}A^{u}(s,t)$ and $A^{d}=\underset{s\in S}{\cup}A^{d}(s)$, where $A^{u}(s,t)$ and $A^{d}(s)$ are the action sets of the user with type $t$ and the defender at state $s$, respectively.
\item ${\sf \dot{p}}$: state transition probability function. Let ${\sf \dot{p}}:S\times \textit{Ty}\times A^{u}\times A^{d}\times S\rightarrow [0,1]$.
\item ${\sf \dot{f}}^u$ and ${\sf \dot{f}}^d$: the immediate payoff functions for the user and the defender, respectively. Let ${\sf \dot{f}}^u: S\times \textit{Ty}\times A^{u}\times A^{d}\rightarrow \mathds{R}$, ${\sf \dot{f}}^d: S\times \textit{Ty}\times A^{u}\times A^{d}\rightarrow \mathds{R}$, where $\mathds{R}$ is the real number.
\end{enumerate}

The PVCCS\textsubscript{G} based security model represents the network as a state transition system.
The processes in PVCCS\textsubscript{G} represent all possible behaviors of the participants at each state, and each state is assigned with a process depicting all possible interactions of the participants. Technically, let $\mathcal{A}=\{{\sf Aces}$, ${\sf Defd}$, ${\sf Tell}_u,$ ${\sf Tell}_d\}$ and $Label=\mathcal{A}\cup \overline{\mathcal{A}}\cup \{{\sf \overline{Log}}\}\cup \{{\sf \overline{Rec}}\}$.
${\it Val}=A^u\cup A^d\cup H\cup \textit{Ty}$, where $H \subseteq \mathds{R}\times\mathds{R}$.
$Act= Act^u\cup Act^d\cup Act^n$, where $Act^u$, $Act^d$ and $Act^n$ are the behavior sets of the user, the defender and the network environment, respectively.
\begin{equation*}
\footnotesize
\begin{aligned}
Act^u=&\{{\sf \overline{Aces}}(u)\mid u\in A^u \}\cup\{{\sf Tell}_u(x)\mid x\in {\it Var}\}\\
Act^d=&\{\overline{{\sf Defd}}(v)\mid v\in A^d \}\cup \{{\sf Tell}_d(x)\mid x\in {\it Var}\}\\
Act^n=&\{{\sf Aces}(x)\mid x\in {\it Var} \}\cup \{{\sf Defd}(x)\mid x\in {\it Var} \}\\
&\cup \{\overline{{\sf Tell}_u}(x)\mid x\in {\it Var} \}
\cup \{\overline{{\sf Tell}_d}(x)\mid x\in {\it Var}\}\\
&\cup \{\overline{{\sf Log}}(x,y)\mid x\in {\it Var}, y\in {\it Var}\}\cup \{\overline{{\sf Rec}}(r)\mid r\in H\}\\
\end{aligned}
\end{equation*}

The processes $G_{\it i}$,  $\textit{pU}_{i}$, $\textit{pD}_{i}$ and $\textit{pN}_{i}$, separately depicting all possible behaviors of the Nature, the user, the defender and the network environment at state $s_{\it i}$, are defined as follows.

\begin{footnotesize}
\begin{align*}
\textit{G}_{i}&\stackrel{\textit{def}}{=}\underset{t\in \textit{Ty}}{\sum} [q]N_{i}(t),~~~N_{\textit{i}}(t)\stackrel{\textit{def}}{=}(\textit{pU}_i(t)|
\textit{pD}_i|\textit{pN}_i(t)) \backslash R\\
\textit{pU}_{i}(t)&\stackrel{\textit{def}}{=}\underset{u\in A^u(s_{\textit{i}},t)}{\sum}\overline{{\sf Aces}}(u).{\sf Tell}_u(y).Nil        \\
\textit{pD}_{i}&\stackrel{\textit{def}}{=}~{\sf Tell}_d(x).\underset{v\in A^d(s_{\textit{i}})}{\sum}\overline{{\sf Defd}}(v).Nil\\
\textit{pN}_{i}(t)&\stackrel{\textit{def}}{=}   ~{\sf Aces}(x).\overline{{\sf Tell}_d}(x).{\sf Defd}(y).\overline{{\sf Tell}_u}(y). Tr_{\textit{i}}(x,y,t)\\
Tr_{\textit{i}}(x,y,t)&\stackrel{\textit{def}}{=}\underset{u\in A^u(s_{\textit{i}},t), v\in A^d(s_{\textit{i}})}{\sum}{\rm {\bf if}}~(x=u,y=v)~ {\rm {\bf then}}\\
&\overline{{\sf Log}}(u,v).\overline{{\sf Rec}}(r^u, r^d).\underset{j\in I}{\sum}[p_{\it ij}]G_{j} ~{\rm {\bf else}}~Nil
\end{align*}
\end{footnotesize}
where $R=
\{{\sf Aces},{\sf Defd},{\sf Tell}_{\it u},{\sf Tell}_{\it d}\}$, $q$ is the probability distribution of type $t$, $p_{\it ij}={\sf \dot{p}}(s_{\textit{i}},t,u,v,s_{\textit{j}})$,
$r^u={\sf \dot{f}}^u(s_{\it i},t,u,v)$, $r^d={\sf \dot{f}}^d(s_{\it i},t,u,v)$.

$G_{\it i}$ means that at state $s_{\it i}$, the Nature presumes the type $t$ with probability $q$.
$N_{\it i}(t)$ means the defender will interact with the user with type $t$ at $s_{\it i}$.
$pU_{\it i}(t)$ means the type $t$ user launches an access request $u$ ($\overline{{\sf Aces}}(u)$), and waits for the responses from the network environment (${\sf Tell}_{\it u}(y)$).
$pD_{\it i}$ means the defender captures some  potential attacks happened (${\sf Tell}_{\it d}(x)$), and then sends a defense instruction $v$ to the network environment ($\overline{{\sf Defd}}(v)$).
$pN_{\it i}(t)$ means the network environment receives an access request from the type $t$ user (${\sf Aces}(x)$), and informs the defender of the request from the user ($\overline{{\sf Tell}_d}(x)$), after receiving a defense instruction from the defender (${\sf Defd}(y)$), the network environment will reply the user with the defensive information ($\overline{{\sf Tell}_u}(y)$).
At last, the network environment generates a log file to record the interaction ($\overline{{\sf Log}}(x,y)$) and evaluate the payoffs for the user and the defender caused by this interaction ($\overline{{\sf Rec}}(r^u, r^d)$), finally the network system evolves to another state with probability $p_{\it ij}$.
Based on process transition rules, we obtain the network state transition graph caused by offensive-defensive interactions.
\subsection{SecModel}
To keep the realistic states,
we abstract the state transition graph via path contraction \cite{reinhard} to a labeled graph named as SecModel.
The vertex set is $V=\underset{s_{\it i}\in S}{\cup}\{G_{\it i}\}$, $G_{\it i}$ is the process assigned to state $s_{\it i}$.
The edge set of $G_{\it i}$ is $E(G_{\it i})$, ranged over by $e_{\it ij}=(G_{\it i},G_{\it j})$ if $G_{\it i}\stackrel{\tau^4[q]}{\rightarrow}\cdot\stackrel{\overline{{\sf Log}}(u,v)}{\rightarrow}\cdot\stackrel{\overline{{\sf Rec}}(r^u,r^d)[p_{\it ij}]}{\rightarrow}G_{\it j}$, and $u\in A^u(s_{\it i},t)$.
The label of edge $e_{\it ij}$ is $L(e_{\textit{ij}})$=$(L_{\textit{Type}}(e_{\textit{ij}})$ , $L_{\textit{TypePr}}(e_{\textit{ij}})$, $L_{\textit{Act}}(e_{\textit{ij}})$,  $L_{\textit{TranP}}(e_{\textit{ij}})$, $L_{\textit{WeiP}}(e_{\it ij}))$. $L_{\textit{Type}}(e_{\textit{ij}})=t$ is the user's type, $L_{\textit{TypePr}}(e_{\textit{ij}})=q$ is type's probability distribution, $L_{\textit{Act}}(e_{\textit{ij}})=(u,v)$ is offensive-defensive action, $L_{\textit{TranP}}(e_{\textit{ij}})=p_{\it ij}$ is the transition probability and $L_{\textit{WeiP}}(e_{\it ij})=(r^u,r^d)$ is the weight pair of this interaction. Later, superscript $u$ and $d$ distinguish the value for the user and the defender.

\section{Analyzing NES on SecModel}
\subsection{Nash Equilibrium Strategy on SecModel}
\begin{definition}
A \textit{t-execution} of $G_{\textit{i}}$ in {\rm SecModel}, denoted as $\pi_{\textit{i}}^{t}$, is a walk (vertices and edges appearing alternately) starting from $G_{\textit{i}}$ and ending with a cycle, on which every vertex's out-degree is 1 and each edge $e$ has label $L_{\textit{Type}}(e)=t$.
\end{definition}

$\pi_{\textit{i}}^{t}[j]$ denotes the subsequence of $\pi_{\textit{i}}^{t}$ starting from $G_{\textit{j}}$ if $G_{\textit{j}}$ is a vertex on $\pi_{\textit{i}}^{t}$.

\begin{definition}
The \textit{payoffs} of the user and the defender on execution $\pi_{\textit{i}}^t$, denoted by $PF^u(\pi_{\textit{i}}^t)$ and $PF^d(\pi_{\textit{i}}^t)$, respectively, are defined as follows:
\begin{align*}
PF^u(\pi_{\textit{i}}^t)&=L^u_{\textit{WeiP}}(e_{\textit{ij}}) +\beta\cdot L_{\textit{TranP}}(e_{\textit{ij}})\cdot PF^u(\pi_{\textit{i}}^t[j])\\
PF^d(\pi_{\textit{i}}^t)&=L^d_{\textit{WeiP}}(e_{\textit{ij}}) +\beta\cdot L_{\textit{TranP}}(e_{\textit{ij}})\cdot PF^d(\pi_{\textit{i}}^t[j])
\end{align*}
where $\beta\in(0,1)$ is a discount factor, $e_{\it ij}=(G_{\it i},G_{\it j})$.
\end{definition}

\begin{definition}
$\pi_{\textit{i}}^t$ is a \textit{t-Nash Equilibrium Execution (t-{\rm NEE})} of $G_{\textit{i}}$ if it satisfies:
\begin{align*}
PF^u(\pi_{\textit{i}}^t)&=\underset{e_{\textit{ij}}\in E(G_{\it i}) }{\max}\{L_{\textit{WeiP}}^u({e_{\textit{ij}}})+ \beta\cdot L_{\textit{TranP}}(e_{\textit{ij}})\cdot PF^u(\pi_{\textit{j}}^t)\} \\
PF^d(\pi_{\textit{i}}^t)&=\underset{e_{\textit{ij}}\in h_{\it i}^t(e)}{\max}\{L_{\textit{WeiP}}^d({e_{\textit{ij}}})+ \beta\cdot L_{\textit{TranP}}(e_{\textit{ij}})\cdot PF^d(\pi_{\textit{j}}^t) \}
\end{align*}
where $e$ is the first edge of $\pi_{\textit{i}}^t$,
$h_{\it i}^t(e)=\{e'\in E(G_{\it i})\mid L_{\it Act}^{u}(e')=L_{\it Act}^{u}(e), L_{\it Type}(e')=L_{\it Type}(e)=t\}$,
$\pi_{\textit{j}}^t$ is the \textit{t-{\rm NEE}} of $G_{\textit{j}}$. It is defined coinductively {\rm \cite{davide07}}.
\end{definition}

\begin{definition}
\textit{Strategy} is a spanning subgraph of {\rm SecModel} satisfying:
 \begin{itemize}
 \item for any $e$, $e'\in E(G_{\it i})$, $G_{\it i}\in V$, $L_{\it Type}(e)\neq L_{\it Type}(e')$;
 \item $\underset{e\in E(G_{\it i})}{\bigcup}L_{\it Type}(e)={\it Ty}$.
 \end{itemize}
\end{definition}

\begin{definition}
\textit{Nash Equilibrium Strategy} ({\rm NES}) is a strategy in which $\forall G_{\textit{i}}$, $\forall t\in {\it Ty}$, any $\textit{t-execution}$ of $G_{\textit{i}}$ is its $\textit{t-{\rm NEE}}$.
\end{definition}

\subsection{Algorithm}
The algorithm we proposed to compute NES on SecModel, denoted as \verb"FindNES()", works as follows:
\begin{enumerate}[1]
\item minimize the network state transition graph by probabilistic bisimulation (function \verb"Minimization"());
\item abstract the minimized graph via path contraction to SecModel (function \verb"Abstraction"() );
\item compute the defender's expected payoff up to his belief on the user's type (function \verb"BayExp"());
\item stratify the model via Tarjan's strongly connected component algorithm \cite{reinhard} and compute the NESs backward inductively (function \verb"AlgNES"()).
\end{enumerate}

Assume the maximum out-degree of each vertex is $M$, $n$ and $m$ denote the size of vertex set and edge set of SecModel, respectively, the complexity of \verb"FindNES()" is $O(n\times M^2+n+m)$.

\verb"Minimization"():
the input is the network state transition graph and the output is the minimized graph. It is a recursive function and works as follows:
\begin{enumerate}[1]
\item for any $G_{\it i}$, $G_{\it j}$, if $\forall t\in {\it Ty}$, $e=(G_{\it i},G_{\it i}')$, $L_{\it Type}(e)=t$, $\exists e'\in E(G_{\it j})$,  $e'=(G_{\it j},G_{\it j}')$, $L(e)=L(e')$ componentwise, label ($G_{\it i}$, $G_{\it j}$) with {\rm Bisim}. \verb"Minimization"($G_{\it i}'$, $G_{\it j}'$);
\item else we label ($G_{\it i}$, $G_{\it j}$) with {\rm NonBisim}, return ${\it false}$;
\item if ($G_{\it i}'$, $G_{\it j}'$) has label {\rm Bisim}, return ${\it true}$.
\end{enumerate}

\verb"Abstraction"():
its input is the minimized graph and the output is SecModel. It works as follows:
\begin{enumerate}[1]
\item pick any two paths of some $G_{\it i},G_{\it j}\in V$, respectively;
\item if they are vertex independent which means they have no common internal vertex, contract each path as a single edge between the endpoints.
    Keep the values transferred by this multi-transition as the edge label;
\item else then these paths are kept intact.
\end{enumerate}

\verb"BayExp"():
the input is $E(G_{\it i})$ of any $G_{\it i}$, and the output is $E(G_{\it i})$ with modified $L_{\textit{WeiP}}^d(e)$.
For any $e\in E(G_{\it i})$, if $L_{\textit{Act}}(e)=(u,v)$, $L_{\textit{type}}(e)=t$, $L_{\textit{TypePr}}(e)=q$, we use Bayesian rule to update $L_{\textit{WeiP}}^d(e)=\underset{t\in {\it Ty}}{\sum}\delta(t\mid u)\cdot f^d(s_{\textit{i}},t,u,v)$, where
$\delta(t\mid u)=\frac{\mathsf{p}(u\mid t)\cdot q} {\underset{t'\in \textit{Ty}}{\sum}\mathsf{p}(u\mid t')\cdot \mathsf{p}(t')}$
where $\mathsf{p}(t')=L_{\textit{TypePr}}(e')$ if  $L_{\textit{Type}}(e')=t'$. $\mathsf{p}(u\mid t)$ is the probability that action $u$ is observed given the user's type $t$.

\verb"AlgNES"():
the input is the minimized SecModel modified by \verb"BayExp"(), and the output are the NESs of SecModel. It works as follows:
\begin{enumerate}[1]
\item stratify the minimized SecModel to an acyclic graph by viewing each SCC as a cluster vertex. $\textit{Leave}$ denotes the one with zero out-degree. $\textit{NonLeave}$ denotes others.
\item find the NESs for all \textit{Leaves} in parallel. The key point of finding NES for each $\textit{Leave}$ is, $\forall t\in \textit{Ty}$, to find a \textit{t-cycle} in this $\textit{Leave}$ which is a $\textit{t-{\rm NEE}}$ of every vertex on it. \textit{t-cycle} is a cycle whose edge $e$ has label $L_{\textit{Type}}(e)=t$.
\item compute NES for any \textit{NonLeave} backward inductively. It follows the method of finite dynamic games for NES \cite{martin}.
\end{enumerate}

The core of \verb"AlgNES"() is how to find NES for SCC. Let $D$ denote each SCC, and $G_{\it i}\in \mathcal{V}(D)$ if $G_{\it i}$ belongs to $D$. It is a value iteration process.
The value function, named as \verb"LocNs"(), is to select some edge $e$ of $G_{\textit{i}}$, for all $G_{\it i}\in \mathcal{V}(D)$, satisfying Nash Equilibrium condition.
The iterated function, named as \verb"RefN"(), is to update the weight pair for each edge of $G_{\textit{i}}$.
We use $L_0(e)$ records the weight pair updated by \verb"BayExp"() and $L_{n}(e)$ is the updated weight pair of $e$ on the nth iteration. A variable vector $Pp_{\it n}(G_{\it i})=[Pp_{\it n}^t(G_{\it i})]_{t \in \textit{Ty}}$ saves the weight pair of $e$, if $e$ is the result of \verb"LocNs()" on the nth iteration, that is $Pp_{\it n}^t(G_{\it i})=L_n(e)$ with $L_{\it Type}(e)=t$.
The value iteration process will terminate if the weight pair value of each edge is unchanged.

In \verb"LocNs()", given a type $t$, edge $e\in E(G_{\it i})$ satisfies Nash Equilibrium condition on the nth iteration ($n\geq 0$) if $e=\arg\underset{e'\in h_{\it i}^t(e)}{\max}L_{\it n}^d(e')$ and $e= \arg\underset{e'\in E(G_{\it i})}{\max}L_{\it n}^u(e')$.

In \verb"RefN()", on the nth iteration ($n\geq 1$), the weight pair of each edge $e\in E(G_{\it i})$ with type $t$ is updated by:
$L_{\it n}(e)=L_{0}(e)+\beta \cdot L_{\textit{TranP}}(e)\cdot Pp_{\textit{n-1}}^t(G_{\textit{j}})$, where $e=(G_{\it i},G_{\it j})$.

\subsubsection{Termination and Correctness}
We need to prove the termination of \verb"AlgNES"().
Inspired by a technique in dynamic programming \cite{vander}\cite{ls}, on the kth iteration, the value function can be formalized as a mapping $\sigma: V\rightarrow \mathds{R}\times \mathds{R}$, $\sigma_{k}(G_{\textit{i}})=Pp_{k}(G_{\textit{i}})$; the iteration function defines a set of vertex $\{G_{\textit{i}}(\sigma_{k-1})\mid G_{\textit{i}}(\sigma_{k-1})$ denotes $G_{\textit{i}}$ with $e_{\textit{ij}}$ whose weight pair is $L_{\it k}(e_{\textit{ij}})=L_{0}(e_{\textit{ij}})+ \beta\cdot L_{\textit{TranP}}(e_{\textit{ij}})\cdot \sigma_{k-1}(G_{\textit{j}})\}$.
Then we have $\sigma_{k+1}(G_{\textit{i}})=Pp_{1}(G_{i}(\sigma_{k}))$. We define a shorthand notation $(T\sigma)(G_{\textit{i}})=Pp_{1}(G_{i}(\sigma))$, that is $T\sigma_{k}=\sigma_{k+1}$. We need to prove $T$ is a contraction.
\begin{Lemma}
\label{provenesterminatedpre}
For any $G_{\textit{i}}\in V$, we have
$$\mid \sigma_{k}(G_{\textit{i}})^u- T\sigma_{k}(G_{\textit{i}})^u\mid~ \leq \underset{e\in E(G_{\textit{i}})}{\max} \mid L_{\it k}^u(e)- L_{\it k+1}^u(e)\mid$$
$$\mid \sigma_{k}(G_{\textit{i}})^d- T\sigma_{k}(G_{\textit{i}})^d\mid ~\leq \underset{e\in E(G_{\textit{i}})}{\max} \mid L_{\it k}^d(e)- L_{\it k+1}^d(e)\mid$$
\end{Lemma}
\begin{proof}
We prove it by contradiction. Assuming without loss of generality, for any $e,e'\in E(G_{\it i})$ with $L_{{\textit Type}}(e)=L_{{\textit Type}}(e')$, if $L_{{\textit WeiP}}^u(e)>L_{{\textit WeiP}}^d(e')$, then $L_{{\textit WeiP}}^d(e)<L_{{\textit WeiP}}^d(e')$. This assumption follows the reality: if the user is an attacker, then the more the damages he causes, the more time the defender spends to normalizing the network; if the user is a regular user, the more requests he sends, the more effort the defender takes to balance the load. Let $\sigma_{k}(G_{\textit{i}})=(L^u_{\it k}(e_{1}), L^d_{\it k}(e_{1}))$ and $T\sigma_{k}(G_{\textit{i}})=(L^u_{\it k+1}(e_{2}), L^d_{\it k+1}(e_{2}))$, where $e_1$, $e_2\in E(G_{\textit{i}})$.
Let $L^u_{\it k}(e_{1})=a$, $L^u_{\it k}(e_{2})=b$, $L^u_{\it k+1}(e_{1})=a'$ and $L^u_{\it k+1}(e_{2})=b'$, where $a$, $a'$, $b$, $b'$ are positive number. We prove the first inequality. Similar to the second one. \\
\textbf{case 1}: $L_{\textit{Act}}^u(e_{1})=L_{\textit{Act}}^u(e_{2})$\\
According to the Nash Equilibrium condition, we have $a<b$ and $b'<a'$. If the first inequality in the lemma does not hold, then we have $\mid a-b'\mid >\mid a-a'\mid$ and $\mid a-b' \mid>\mid b-b'\mid$, then we get $(b'-a')(b'+a')>2a(b'-a')$ and $(a-b)(a+b)>2b'(a-b)$ which deduce $a-b>a'-b'$, contradiction.\\
\textbf{case 2}: $L_{\textit{Act}}^u(e_{1})\neq L_{\textit{Act}}^u(e_{2})$.
Let's define two conditions:
Cond 1: $Pp_{n}(G_{\it i})=L_{\it n}(e_{2})$;
Cond 2: $Pp_{n+1}(G_{\it i})=L_{\it n+1}(e_{1})$\\
\textbf{$\textit{case 2.1}$}: both Cond 1 and Cond 2\\
Then $a>b$ and $b'>a'$. If $\mid a-b'\mid >\mid a-a'\mid$ and $\mid a-b' \mid>\mid b-b'\mid$, then $b-a>b'-a'$, contradiction. \\
\textbf{$\textit{case 2.2}$}: not Cond 2 but Cond 1 \\
Then $\exists e'$ with $L_{Act}^u(e')=L_{Act}^u(e_{1})$. Let $L_{\it k}^u(e')=c$, $L_{\it k+1}^u(e')=c'$, then $c>a>b$, $a'>c'$, $b'>c'$. If $\mid a-b'\mid >\mid a-a'\mid$ and $\mid a-b' \mid>\mid b-b'\mid$, then $(b'-a')(b'+a')>2a(b'-a')$, $a+b>2b'$. If $b'>a'>c'$, contradiction; If $c'<b'<a'$, then $2b'<b'+a'<2a$,  $2c>a+b>2b'>2c'$, so $b'<a$ and $c>c'$. If $\mid a-b' \mid>\mid c-c'\mid$, then $a-c>b'-c'$; If $b'=a'$, contradiction. \\
\textbf{$\textit{case 2.3}$}: not Cond 1 but Cond 2\\
Proof is similar to \textbf{$\textit{case 2.2}$}.\\
\textbf{$\textit{case 2.4}$}: neither Cond 1 nor Cond 2\\
Then $\exists e', e''$, $L_{Act}^u(e')=L_{Act}^u(e_{1})$, $L_{Act}^u(e'')=L_{Act}^u(e_{2})$. Let $L^u_{\it k}(e')=c$, $L^u_{\it k+1}(e')=c'$, $L^u_{\it k}(e'')=d$, $L^u_{\it k+1}(e'')=d'$, then $d<a<c$, $d<b$, $a'>c'$, $c'<b'<d'$. If $\mid a-b'\mid >\mid a-a'\mid$, $\mid a-b' \mid>\mid b-b'\mid$, then $(b'-a')(b'+a')>2a(b'-a')$, $(a-b)(a+b)>2b'(a-b)$. If $a>b$ and $a'>b'$, then $c'<c$ and $a>b'$. If $\mid a-b'\mid >\mid c-c'\mid$, then $a-c>b'-c'$, contradiction; If $a<b$ and $a'>b'$ or $a>b$ and $a'<b'$, contradiction; If $a<b$, $a'<b'$, then $d'>d$, $a<b'$. If $\mid a-b'\mid >\mid d-d'\mid$, then $b'-d'>a-d$, contradiction.
\end{proof}

\begin{Lemma}
\label{provenesterminated}
$T$ is a contraction, i.e. $T$ has a fixed point.
\end{Lemma}
\begin{proof}
For any real vector $\overrightarrow{x}\in \mathds{R}^J$, $J$ is an index set, let $\mid \mid \overrightarrow{x}\mid\mid_{\infty}=\max_{j}|x_{j}|$. According to Lemma \ref{provenesterminatedpre}, then
\begin{align*}
\mid\mid T\sigma_{k+1}^u-T\sigma_{k}^u\mid\mid_{\infty}&=\underset{G_{i}\in V}{\max}\mid T\sigma_{k+1}(G_{i})^u-T\sigma_{k}(G_{i})^u \mid\\
&\leq\underset{G_{j}\in V}{\max}\beta\cdot \mid   \sigma_{k+1}(G_{\textit{j}})^u-\sigma_{k}(G_{\textit{j}})^u \mid\\
&=\beta\cdot\mid\mid \sigma_{k+1}^u-\sigma_{k}^u\mid\mid_{\infty}
\end{align*}
similar proof for $\mid\mid T\sigma_{k+1}^d-T\sigma_{k}^d\mid\mid_{\infty}\leq\beta\cdot\mid\mid \sigma_{k+1}^d-\sigma_{k}^d\mid\mid_{\infty}$. As $\beta\in (0,1)$ and regardless of the initial value function $\sigma_{0}$, sequence $\sigma_{k}$ converges to a unique limit $\sigma^*$ with $T\sigma^*=\sigma^*$.
\end{proof}

\begin{Theorem}
\label{correctness}
\verb"FindNES()" finds all {\rm NESs} of {\rm SecModel}.
\end{Theorem}
\begin{proof}
We prove:
1. \verb"FindNES"() is terminated. It is trivial
by Lemma \ref{provenesterminated};
2. \verb"FindNES"() finds all NESs. We prove it backward inductively. If vertex $D$ is a $Leave$, for $\forall G_{i}\in \mathcal{V}(D)$, $\forall t\in {\it Ty}$, assuming $\pi_{\textit{i}}^t$ whose first edge is $e$ (written as $\pi_{\textit{i}}^t(e)$) is the execution obtained by \verb"FindNES"(). If $\pi_{\textit{i}}^t(e)$ is not \textrm{NEE} of $G_{\textit{i}}$, then there is $\pi_{\textit{i}}^{t}(e')$ with
$PF^d(\pi_{\textit{i}}^{t}(e'))> PF^d(\pi_{\textit{i}}^{t}(e))$, $L_{\textit{Act}}^u(e')= L_{\textit{Act}}^u(e)$ or
$PF^u(\pi_{\textit{i}}^{t}(e'))> PF^u(\pi_{\textit{i}}^{t}(e))$, $e'=\arg\underset{e''\in E(G_{\it i})}{\max} PF^d(\pi_{\textit{i}}^{t}(e''))$, contradiction.
If $D$ is a ${\it NonLeave}$, according to the definition of NES, trivial.
\end{proof}

\section{Applications}
The efficiency of our approach is illustrated by two detailed examples.
All our experiments were carried out on 2.53 GHz i5 core computer with 4G RAM.

\subsection{Defense for DDoS Attacks}
This example is referred from the literature \cite{harkeerat} and it is usually modeled via incomplete information games.
It shows how to protect a network system from being attacked by DDoS \cite{david}.
In this case, the user can be a legitimate user who sends packages with normal service request or a zombie user who sends packages with fake IPs. The defender cannot distinct the rogue flow from the legitimate flow, so the defender's challenge is to determine the optimal firewall settings to block rogue traffics while allowing legitimate ones.
The legitimate user will try to make the most of the bandwidth to speed up his request to the server, while the malicious user will attempt to find the most effective sending rate and botnet size to exhaust the bandwidth without being detected.
It's necessary to model for all the possible interactions under different settings and find the most effective one.

There is one state $s_1$ in this example, so the defender will update continuously the judgement for the user's type under the  interactions repeatedly happened.
The type set {\it Ty}=$\{$ {\it Zombie}, {\it Regular}$\}$.
The Nature presumes these types with the same probability.
$A^u$=\{($r_{\it u}$, $m_{\it u}$) $\mid$ $r_{\it u}\in \mathds{R}^+$, $m_{\it u}\in \mathds{N}$\}.
For the {\it Zombie} user, $r_{\it u}$ and $m_{\it u}$ denote the zombie flow rate and the botnet size, while for the {\it Regular} user, $r_{\it u}$ and $m_{\it u}$ denote the request flow rate and the number of the request sent at a time.
$A^d$=\{${\it Mp}$ $\mid$ ${\it Mp}\in \mathds{R}^+$\}, ${\it Mp}$ denotes the parameter for the firewall's dropping rate.
We assume there are already $n$ legitimate flow with flow rate $r_{\it l}$ to be sent and the payoff of the defender is equivalent in the absolute value to the user's payoff.

For the \textit{Zombie} user, his immediate payoff is measured by the bandwidth occupied by the zombie flow ($f_{\it z}$), lost bandwidth  for the regular flow ($f_{\it l}$), and the cost to control the botnet ($f_{\it c}$). So we have ${\sf \dot{f}}^{\it u}(s_1,{\it Zombie}, m_{u},r_{u},{\it Mp})= f_{\it z}+f_{\it l}-f_{\it c}$. Let $\omega_{\it c}$ be a given coefficient, $f_{\it c}=m_{\it u}\times \omega_{\it c}$.

\begin{footnotesize}
\begin{align*}
f_{\it z}=& \frac{B\times m_{u}\times r_{\it u}^{'}}{m_{u}\times r_{\it u}^{'}+n\times r_{\it l}^{'}},~~~~
f_{\it l}= (1-\frac{B}{m_{u}\times r_{\it u}^{'}+n\times r_{\it l}^{'}})\times n\times r_{\it l}'
\end{align*}
\end{footnotesize}
$r_{\it u}'$ and $r_{\it l}'$ mean the flow rate considering the firewall's dropping rate modeled by a function $F(x)$ \cite{harkeerat}. Let $\rho$ is an empirically given scaling factor, $B$ is the network bandwidth.
\begin{footnotesize}
\begin{align*}
F(x)=&\frac{1}{1+e^{-\rho\times \frac{x-{\it Mp}}{B}}},~r_{\it l}^{'}=r_{\it l}\times (1-F(r_{\it l})),~
r_{\it u}^{'}=r_{\it u}\times (1-F(r_{\it u}))
\end{align*}
\end{footnotesize}

For the \textit{Regular} user, his payoff is measured by the number of his request flow arriving at the server.
$${\sf \dot{f}}^{\it u}(s_1, {\it Regular}, m_{\it u}, r_{\it u},{\it Mp})= m_{\it u}\times r_{\it u}' + n\times r_{\it l}'$$

In this example, the SecModel is a directed graph with parameter labels, so we use {\it MATLAB} to accomplish our algorithm and find NES.
We set $B=2000$Mbps, $n=20$, $r_{\it l}=60$, $\rho=-20$, $\omega_{\it c}=10$ and assume $m_{\it u}\times r_{\it u}=800$ for the $Regular$ user and $m_{\it u}\times r_{\it u}=5000$ for the $Zombie$ user.
The results obtained see Figure \ref{nes3} and Figure \ref{nes4}, respectively.

Figure \ref{nes3} shows the \textit{Regular} user will send 8 more flows at a time with flow rate 100 to access the server. The defender will set the midpoint of firewall to be 228.8.
In this setting, the legitimate user will make the most of the bandwidth (almost 1639.84Mbps) and the drop rate of the firewall is 0.2162 which will allow most of the flow to pass.

Figure \ref{nes4} shows the \textit{Zombie} user will set botnet size as 20 and the sending rate as 250.
The defender will set the firewall midpoint as 322.
In this setting, the zombie flows will exhaust the bandwidth ($f_z=$1500.83Mbps, $f_l=$619.35Mbps) and the drop rate is 0.3274 which could drop more flows to prevent the attack to some extent.
\begin{figure}[h]
\centering
\includegraphics[width=5cm, height=4cm]{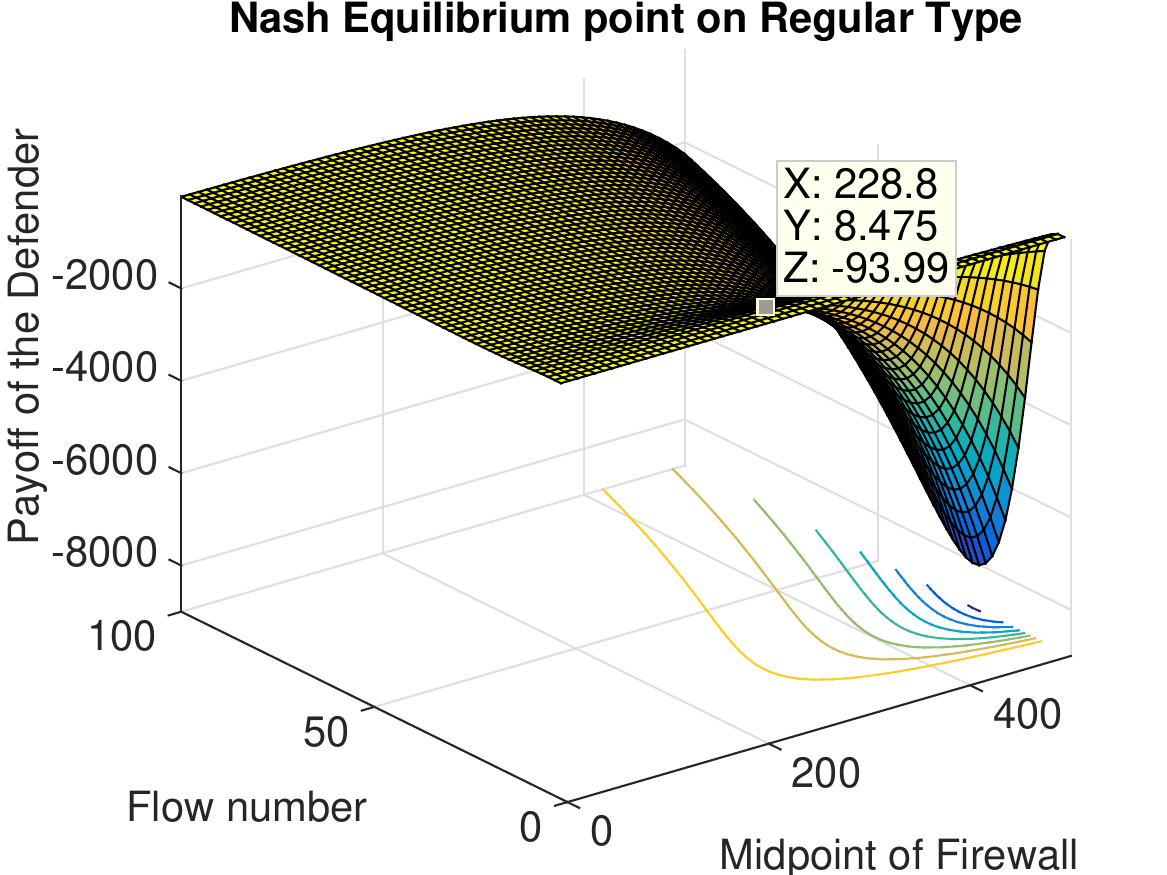}
\caption{NES on Regular type}
\label{nes3}
\end{figure}
\begin{figure}[h]
\centering
\includegraphics[width=5cm, height=4cm]{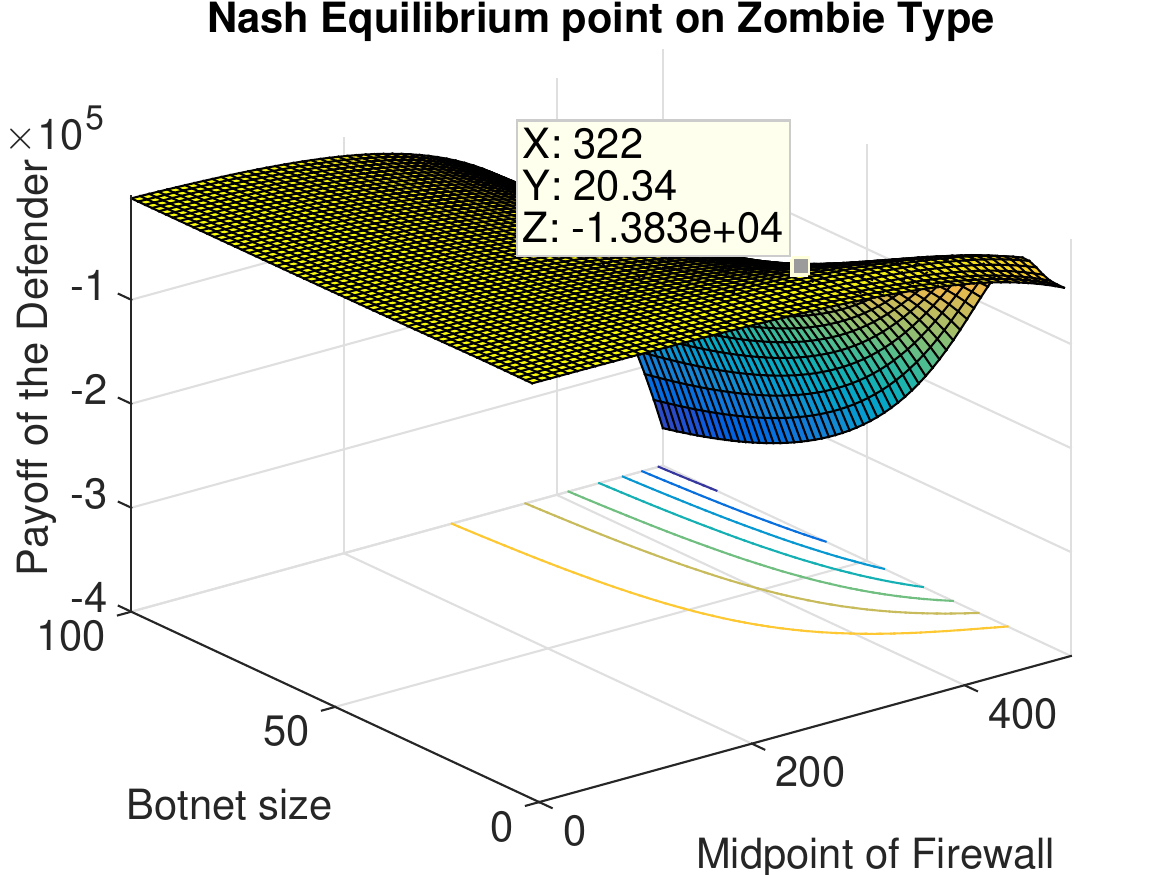}
\caption{NES on Zombie type}
\label{nes4}
\end{figure}

\subsection{Campus Network Defense}
This example is referred from the literature \cite{klye} and is usually modeled via complete information games.
It shows a campus network connected to the Internet (see Figure \ref{example1}), and the user is an attacker who tries to steal or damage some private information data.
It is necessary to find some effective defense deployment in advance by analyzing all possible offensive-defensive interactions.
The type set $\textit{Ty}=\{\textit{Malicious}\}$.
There are 18 states in this example given in Table \ref{state}.
$A^{u}$ and $A^{d}$ are shown in Table \ref{attckeraction} and \ref{defenderaction}, respectively.
We use symbolic number to represent corresponding actions, $\cdot$ means any action.
The transition probability ${\sf \dot{p}}$ is given in Table \ref{statetransition}.
The immediate payoff pair (${\sf \dot{f}}^u$, ${\sf \dot{f}}^d$) is shown in Table \ref{rewardcost}, where ${\sf \dot{f}}^u(s_{\it i})$ and ${\sf \dot{f}}^d(s_{\it i})$ are matrixes with $A^{u}(s_{\it i})$ as columns and $A^{d}(s_{\it i})$ as rows.

The model can be minimized as $s_{13}\sim s_{15}$, $s_{14}\sim s_{16}$ and $s_{17}\sim s_{18}$. Its SecModel sees Figure \ref{ComModelcasestudy}.
Two NESs obtained see Figure \ref{nash1} and Figure \ref{nash2}, respectively.
\begin{figure}[h]
\centering
\includegraphics[width=5cm, height=2.5cm]{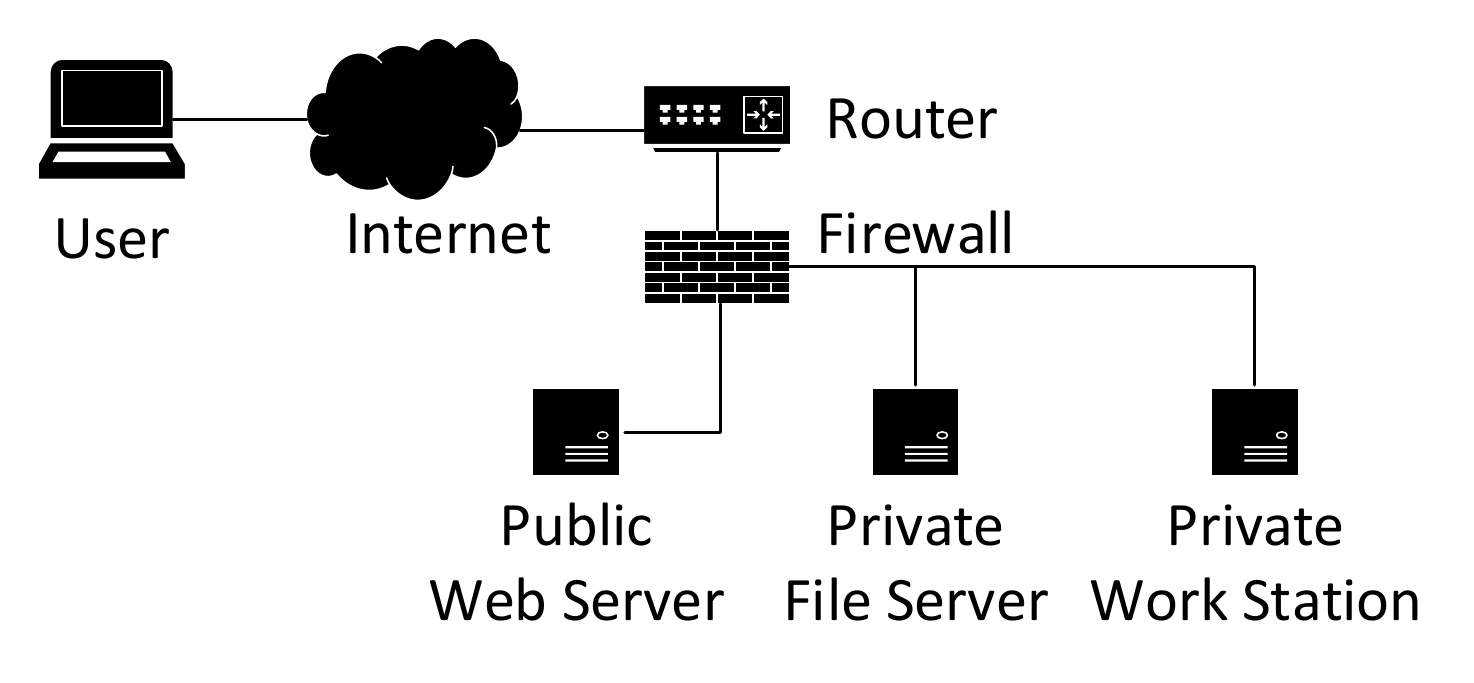}
\caption{Campus Network}
\label{example1}
\end{figure}
\begin{figure}[h]
\centering
\includegraphics[width=9cm, height=6cm]{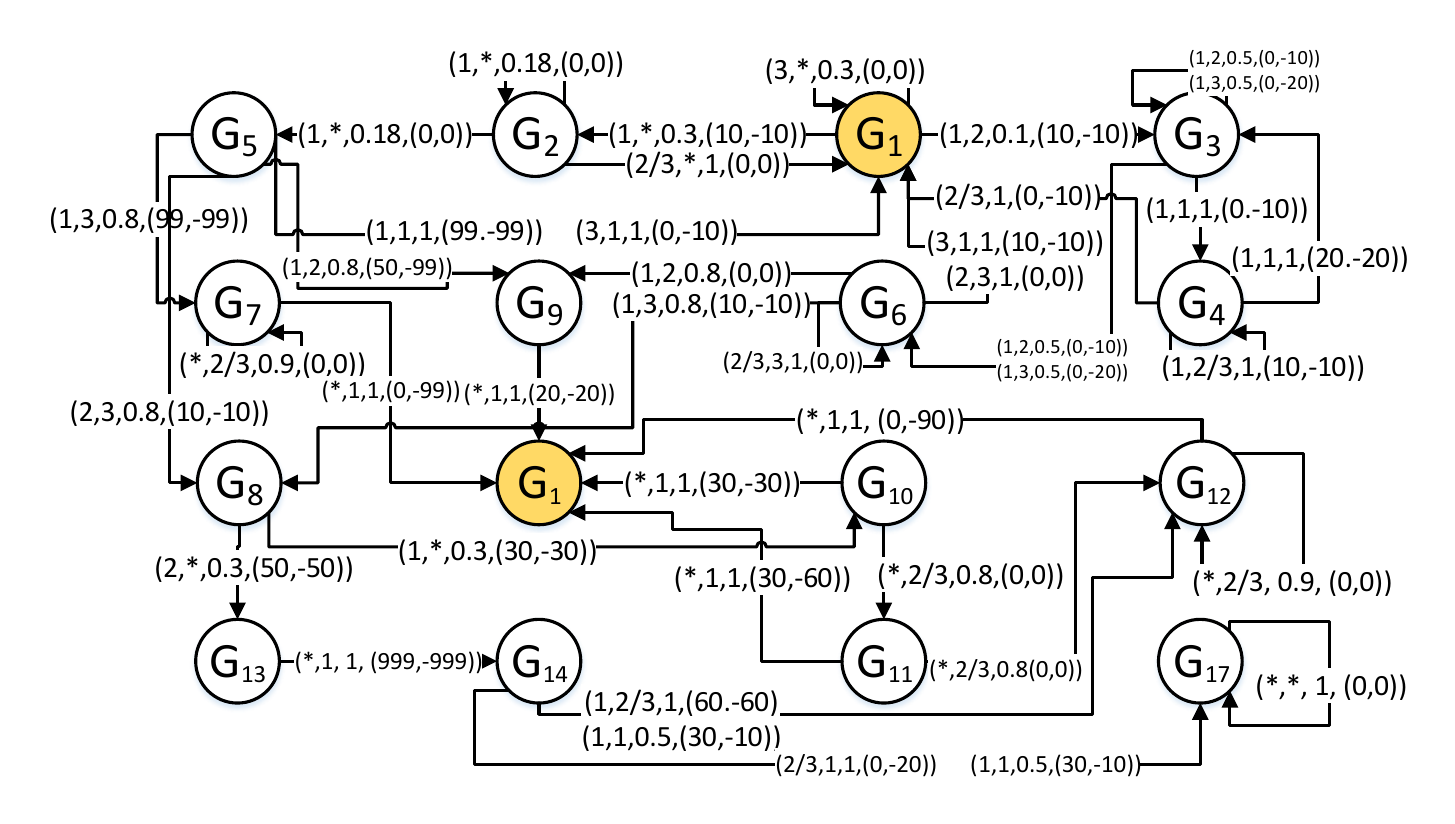}
\caption{SecModel of Campus Network}
\label{ComModelcasestudy}
\end{figure}
\begin{figure}[h]
\centering
\includegraphics[width=6.3cm, height=3.8cm]{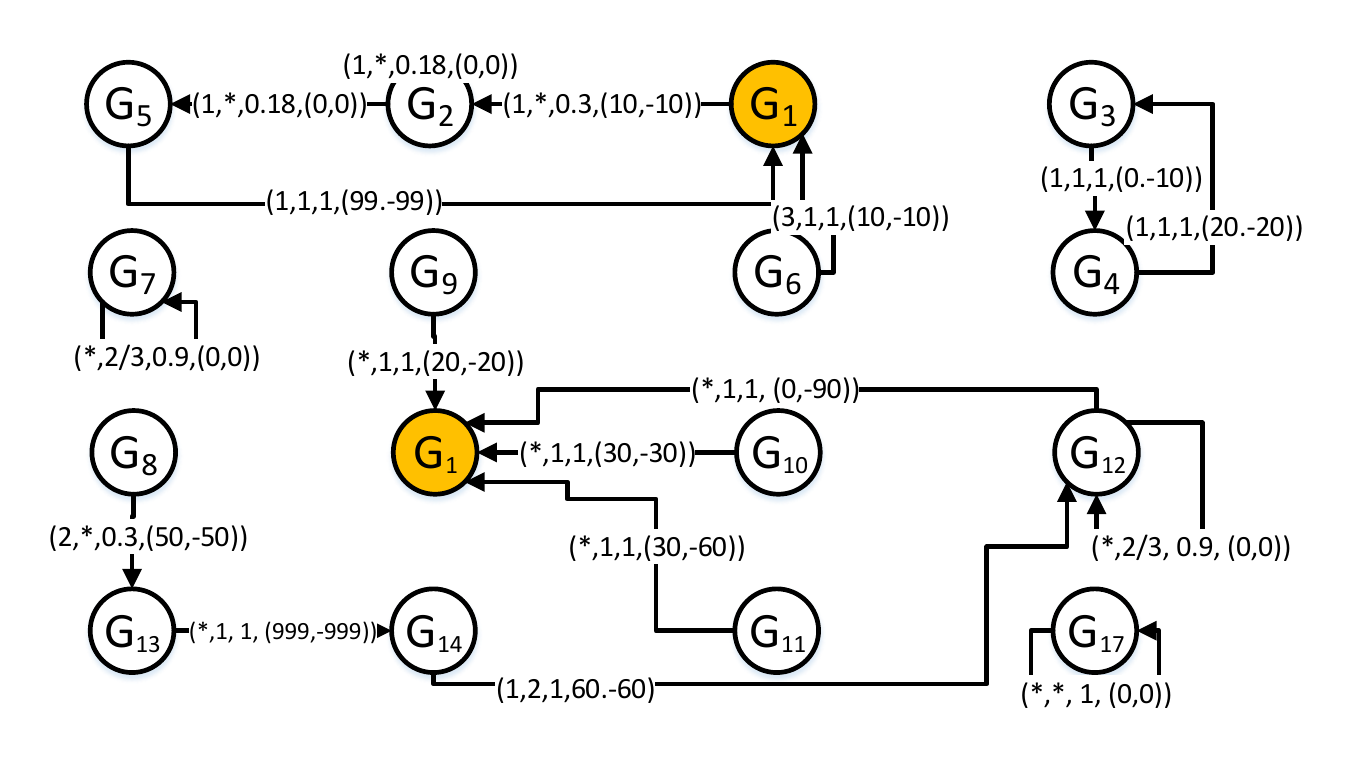}
\caption{The 1st NES}
\label{nash1}
\end{figure}
\begin{figure}[h!t]
\centering
\includegraphics[width=6.3cm, height=3.8cm]{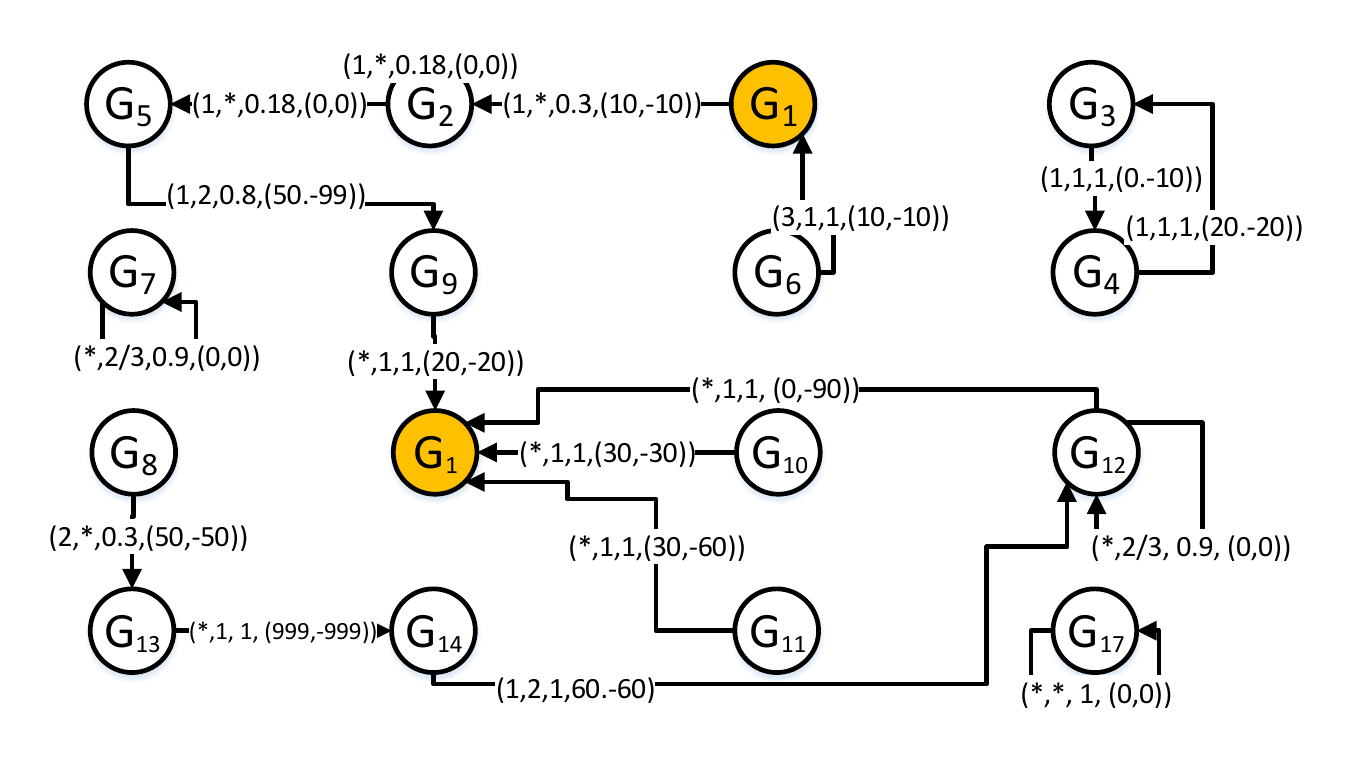}
\caption{The 2nd NES }
\label{nash2}
\end{figure}

The results are largely similar except for a slight difference at $s_{5}$. The first NES tells that for the attacker, even though installing a sniffer may allow him to crack a root password and eventually capture the data he wants, there is also the possibility that the defender will detect his presence and take preventive measures.
He is thus able to do more damages if he simply defaces the web site and leaves.
While for the defender, he should immediately remove the compromised account and restart httpd rather than continue to compete with the attacker.
The second NES shows that the defender should install a sniffer detector. This action can help the defender to further observe the attacker's final object before eventually removing the sniffer program and the compromised account.

Compared with the results obtained by game-theoretic approach\cite{klye}, we filter the invalid NES.
The invalid NES shows at $s_6$ the attacker will install a sniffer and the defender will remove the compromised account and restart ftpd.
However, there is no state transition based on this interaction, so it will never happen in the real if the players are rational.

\section{Conclusion}
We proposed an algebraic model based on a probabilistic extension of the  value-passing CCS, to model and analyze network security scenarios usually modeled via complete or incomplete information games.
Using the algorithm proposed, we computed multiple Nash Equilibria Strategies automatically.
The efficiency and effectiveness of our approach have been illustrated by two detailed applications.
We claimed and proved that our approach can be regarded as a uniform framework for modeling and analyzing the different network scenarios.

In the future, we wish to develop a security model based on CCS for Trees \cite{thomas} to analyze effective defense mechanisms under security scenarios with multiple users and defenders.
\section*{Acknowledgment}
This work has been partly funded by the French-Chinese project Locali (NSFC 61161130530 and ANR-11-IS02-0002) and by the Chinese National Basic Research Program (973) Grant No. 2014CB34030.


\bibliographystyle{IEEEtran}
\bibliography{reference}

\appendix

As the limited space, we show partial experimental data of the second case study.
\begin{table}[h]
\scriptsize
\centering
\caption{\label{state}Part of State set}
\newcounter{Rownumber}
\newcommand{\Rown}{\stepcounter{Rownumber}\theRownumber}
\begin{tabular}{cc}\hline
State number & State name\\\hline
\Rown        & ${\it Normal\_operation}$\\
\Rown        & ${\it Httpd\_attacked}$\\
\Rown        & ${\it Ftp\_attacked}$\\
\Rown        & ${\it Ftpd\_attacked_detector}$\\
\Rown        & ${\it Httpd\_hacked}$\\
\Rown        & ${\it Ftpd\_hacked}$\\
\hline
\end{tabular}
\end{table}

\begin{table}[h]
\scriptsize
\centering
\caption{\label{attckeraction} Part of Attacker's action set}
\newcounter{Rownumberb}
\newcommand{\Rownb}{\stepcounter{Rownumberb}\theRownumberb}
\begin{tabular}{cccc}\hline
$State~no.\backslash$  &  1      &  2       & 3         \\
Action no.             &         &          &           \\\hline
\Rownb             &${\it Attack\_httpd}$ &${\it Attack\_ftpd}$ &$\phi$\\
\Rownb             &${\it Continue\_attacking}$  &$\phi$ &  $\phi$\\
\Rownb             &${\it Continue\_attacking}$  &$\phi$ &  $\phi$\\
\Rownb             &${\it Continue\_attacking}$  &$\phi$ &  $\phi$\\
\Rownb       &${\it Deface\_website}$  &${\it Install\_sniffer}$&$\phi$\\
\Rownb       &${\it Install\_sniffer}$ &$\phi$     &   $\phi$\\
\hline
\end{tabular}
\end{table}

\begin{table}[h]
\caption{\label{defenderaction} Part of Defender's action set}
\centering
\scriptsize
\newcounter{Rownumbera}
\newcommand{\Rowna}{\stepcounter{Rownumbera}\theRownumbera}
\renewcommand{\multirowsetup}{\centering}
\begin{tabular}{cccc}\hline
$State~no.\backslash$  &  1      &  2       & 3         \\
Action no.             &         &          &           \\\hline
\Rowna                 & $\phi$  &  $\phi$  & $\phi$    \\
\Rowna                 & $\phi$  &  $\phi$  & $\phi$    \\
\Rowna                 & ${\it InstallSnifferDetctor}$ &$\phi$ &$\phi$\\
\Rowna                 & ${\it RemoveSnifferDetctor}$  &$\phi$ &$\phi$\\
\Rowna    &${\it RemoveComAccount}$ &${\it InstallSnifferDetector}$  &   $\phi$\\
\Rowna    & ${\it RemoveComAccount}$ &${\it InstallSnifferDetector}$  &   $\phi$\\
\hline
\end{tabular}
\end{table}

\begin{table}[h]
\centering
\scriptsize
\caption{\label{statetransition} Part of Transition probabilities}
\begin{tabular}{lll}\hline
$\begin{bf}State~1\end{bf}$         &$\begin{bf}State~2\end{bf}$        &$\begin{bf}State~3\end{bf}$\\
${\sf \dot{p}}(s_1,1,\cdot,s_2)=1/3$                        &${\sf \dot{p}}(s_2,1,\cdot,s_2)=0.5/3$                         &${\sf \dot{p}}(s_3,1,2, s_3)=0.5$\\
${\sf \dot{p}}(s_1,1,2, s_3)=1/3$                          &${\sf \dot{p}}(s_2,1,\cdot, s_5)=0.5/3$                         &${\sf \dot{p}}(s_3,1,3, s_3)=0.5$\\
${\sf \dot{p}}(s_1,3,\cdot, s_1)=1/3$                      &${\sf \dot{p}}(s_2,2,\cdot, s_1)=1$                             &${\sf \dot{p}}(s_3,1,2, s_6)=0.5$\\
$ $                &${\sf \dot{p}}(s_2,3,\cdot, s_1)=1$             &${\sf \dot{p}}(s_3,1,3, s_6)=0.5$\\
$ $                &$ $             &${\sf \dot{p}}(s_3,1,1, s_4)=1$\\
$\begin{bf}State~4\end{bf}$         &$\begin{bf}State~5\end{bf}$        &$\begin{bf}State~6\end{bf}$\\
${\sf \dot{p}}(s_4,2,1, s_1)=1$                &${\sf \dot{p}}(s_5,1,3,s_7)=0.8$             &${\sf \dot{p}}(s_6,1,3, s_8)=0.8$\\
${\sf \dot{p}}(s_4,3,1, s_1)=1$                &${\sf \dot{p}}(s_5,2,3, s_8)=0.8$             &${\sf \dot{p}}(s_6,1,2, s_9)=0.8$\\
${\sf \dot{p}}(s_4,1,1, s_3)=1$                &${\sf \dot{p}}(s_5,1,2, s_9)=0.8$             &${\sf \dot{p}}(s_6,2,3 ,s_1)=1$\\
${\sf \dot{p}}(s_4,1,2, s_4)=1$                &${\sf \dot{p}}(s_5,3,1, s_1)=1$             &${\sf \dot{p}}(s_6,3,1, s_1)=1$\\
${\sf \dot{p}}(s_4,1,3, s_4)=1$                &${\sf \dot{p}}(s_5,1,1, s_1)=1$             &${\sf \dot{p}}(s_6,2,3 ,s_6)=1$\\
$ $                &$ $             &${\sf \dot{p}}(s_6,3,3, s_6)=1$\\
\hline
\end{tabular}
\end{table}

\begin{table}[h]
\scriptsize
\centering
\caption{\label{rewardcost} Part of Weight pair of each transition}
\begin{tabular}{ll}\hline
${\sf \dot{f}}^u(s_1)=\begin{bmatrix}10&10&10 \\10&10&10 \\0&0&0\end{bmatrix}$           &${\sf \dot{f}}^d(s_1)=-{\sf \dot{f}}^u(s_1)$  \\
${\sf \dot{f}}^u(s_2)=\begin{bmatrix}0&0&0 \\0&0&0 \\0&0&0\end{bmatrix}$                 &${\sf \dot{f}}^d(s_2)={\sf \dot{f}}^u(s_2)$   \\
${\sf \dot{f}}^u(s_3)=\begin{bmatrix}0&0&0 \\0&0&0 \\0&0&0\end{bmatrix}$                 &${\sf \dot{f}}^d(s_3)=\begin{bmatrix}-10&-10&-20 \\-10&-10&0 \\-10&-10&0\end{bmatrix}$\\
${\sf \dot{f}}^u(s_4)=\begin{bmatrix}20&10&10 \\0&0&0 \\0&0&0\end{bmatrix}$              &${\sf \dot{f}}^d(s_4)=\begin{bmatrix}-20&-10&-10 \\-10&0&0 \\-10&0&0\end{bmatrix}$\\
${\sf \dot{f}}^u(s_5)=\begin{bmatrix}99&50&99 \\10&0&10 \\0&10&0\end{bmatrix}$          &${\sf \dot{f}}^d(s_5)=\begin{bmatrix}-99&-99&-99 \\10&10&-10 \\-10&-10&0\end{bmatrix}$\\
${\sf \dot{f}}^u(s_6)=\begin{bmatrix}0&0&10 \\10&0&0 \\10&0&0\end{bmatrix}$          &${\sf \dot{f}}^d(s_6)=-{\sf \dot{f}}^u(s_6)$\\
\hline
\end{tabular}
\end{table}

\end{document}